\documentclass[12pt,a4paper]{amsart}
\usepackage[utf8]{inputenc}
\usepackage[T2A]{fontenc}
\usepackage{amsmath,amssymb,amsthm}
\usepackage[a4paper,hmargin=2.5cm,vmargin=2.5cm]{geometry}
\usepackage[english]{babel}
\usepackage{graphicx}
\usepackage{forest}
\usepackage{tikz-qtree}
\usepackage{algorithm}
\usepackage{algpseudocode}
\usepackage{tikz-cd}
\usepackage{MnSymbol}
\usepackage{enumitem}

\newtheorem{theorem}{Theorem}[section]
\newtheorem{lemma}[theorem]{Lemma}
\newtheorem{prop}[theorem]{Proposition}
\newtheorem{remk}[theorem]{Remark}

\newtheorem{defn}[theorem]{Definition}

\newtheorem{corr}[theorem]{Corollary}

\newcommand{\eqv}{\doteq}
\newcommand{\Wset}{W}
\newcommand{\spann}{\mathrm{span}}

\newcommand{\sbnd}{\,\ldots\,}
\newcommand{\nbas}{S}
\newcommand{\abas}{\Lambda}
\newcommand{\nebas}{\Wset_2\setminus \{\varepsilon\}}
\newcommand{\nebasr}{\Wset_r\setminus \{\varepsilon\}}

\title[A simple algorithm for checking equivalence of counting functions]{A simple algorithm for checking equivalence \\ of counting functions on free monoids}

\author{Petr Kiyashko}
\address{\parbox{\linewidth}{Moscow Institute of Physics and Technology, Institutskii per. 9,\\ 
141701 Dolgoprudny, Moscow region, Russia}}
\email{pskiyashko@phystech.edu}

\author{Alexey Talambutsa}
\address{\parbox{\linewidth}{Steklov Mathematical Institute of RAS, 8 Gubkina St., 119991 Moscow, Russia\\
HSE University, Laboratory of Theoretical Computer Science, \\ 11 Pokrovsky Blvd., 109028 Moscow, Russia}}
\email[]{altal@mi-ras.ru }


\begin{document}

\begin{abstract}
    In this note we propose a new algorithm for checking whether two counting
    functions on a free monoid $M_r$ of rank $r$ are equivalent modulo a bounded function. The previously known algorithm has time complexity $O(n)$ for all ranks $r>2$, but for $r=2$ it was estimated only to be $O(n^2)$. We apply a new approach based on the explicit basis expansion and summation of weighted rectangles, which allows us to construct a much simpler algorithm with time complexity $O(n)$ for any $r\geq 2$. We work in the multi-tape Turing machine model with nonconstant-time arithmetic operations.
\end{abstract}

\maketitle

\section{Introduction}
Let $A_r = \{a_1, a_2, \dots, a_r\}$  be an alphabet consisting of $r$ letters, where $r\geq 2$. The free monoid $M_r$ of rank $r$ can be defined as the set of all finite words over $A_r$ including the empty word $\varepsilon$ with the standard operation of word concatenation. The length of a word $w$ is denoted by $|w|$. A word $v\in M_r$ is a \emph{subword} of the word $u\in M_r$ if there exist $x,y\in M_r$ such that $u=xvy$, and a product decomposition of this form constitutes \emph{an occurrence} of $v$ in $u$. Two occurrences $u=x_1 v y_1$ and $u=x_2 v y_2$ \emph{intersect} if $||x_1|-|x_2||\leq |v|$.

Algorithms that find occurrences of a given subword (or a set of subwords) in another word are ubiquitous in modern data processing. A number of such algorithms working in linear time have been developed since the 1970s, among which the two most famous are probably the Knuth--Morris--Pratt and Aho--Corasick algorithms (for details see \cite{cro-ryt}). In these algorithms the existence of a subword $v$ in the word $w$ is checked by a precomputed finite automaton, which in fact describes the (prefix) subword structure for the word $v$. At the same time, the study of subwords constitutes the central topic in word combinatorics, an area of theoretical computer science that is closely related to the combinatorial group, semigroup, and ring theory (see \cite{loth} and \cite{sapir}). 

A notable application of subword search combinatorics in group theory is a construction of Brooks counting quasimorphisms on the free group $\mathbb F_k$, which enabled the study of many questions on the second bounded cohomology group $H^2_b(\mathbb F_k,\mathbb R)$ of the free group. This approach was later generalized to many other important classes of groups (see \cite{brooks,grigorch,eps-fuji,bes-fuj}). A counting quasimorphism on a free group is defined as a symmetrization of the counting function on a free monoid; this function is the following quantitative generalization of the subword occurrence function. 

For any fixed word $v\in M_r$ we define the \emph{elementary counting function} $\rho_v:M_r\to \mathbb Z$. The value $\rho_v(w)$ for a word $w \in M_r$ equals the number of all (possibly intersecting) occurrences of the subword $v$ in the word $w$.
A \emph{counting function} on $M_r$ is a finite linear combination of the form
\begin{equation}
\label{counting_function}
f=\sum_{i=0}^k \alpha_i\rho_{w_i}.
\end{equation}
Here, all words $w_i$ are elements of a fixed monoid $M_r$ and the coefficients $\alpha_i$ are usually taken from the field of real numbers $\mathbb R$ or the ring of integers $\mathbb{Z}$. We see that in the first case the counting functions on $M_r$ form a linear space, which we denote by ${\mathcal C}_r$, and in the second case the counting functions form a $\mathbb Z$-module inside ${\mathcal C}_r$. 

Two counting functions $f,g\in {\mathcal C}_r$ are considered equivalent if their difference $f-g$ is a bounded function on $M_r$, i.e., there exists some constant $c>0$ such that for any word $w\in M_r$ it holds that
\[|f(w)-g(w)|\leq c.\]
We will use the symbol $\eqv$ to denote this equivalence relation\footnote{Notably, as was shown in \cite[Corollary A.5]{ggd}, the equivalence $f\eqv g$ takes place if and only if the counting functions $f$ and $g$ coincide for counting in the cyclic words, i.e., when the occurrences of the words $w_i$ in the sum \eqref{counting_function} are taken from a cyclic word $w$ in the alphabet $A_r$.}. In addition, we adopt the standard notation $\ell^{\infty}(M_r)$ for the space of bounded functions on the monoid $M_r$. 

\smallskip

Clearly, the equivalence classes in $\mathcal C_r$ form a quotient linear space  
\[\widehat{\mathcal C}_r=\mathcal{C}_r/(\ell^{\infty}(M_r)\cap \mathcal{C}_r).\] 
The combinatorial study of this space in \cite{ggd} enabled to construct a polynomial-time algorithm for checking whether a given counting function is bounded, which was then extended to an algorithm for fast summation of counting quasimorphisms in $H^2_b(\mathbb F_k,\mathbb R)$. For integer coefficients and free group $\mathbb F_k$ with $k\geq 2$ or free monoids $M_r$ with $r\geq 3$ these algorithms were shown in \cite{ht-sbornik} to have complexity $O(n)$. For rational coefficients in the same cases, one has complexity $O(n\log n)$. The main goal of this paper is to show that these complexities can also be achieved in the case of $r=2$. We do this by proposing a new algorithm, which uses a new approach to studying the space $\widehat{\mathcal C}_r$ from \cite{alt}.

\smallskip

To specify the input of the algorithm in the case of integer coefficients we present an element of $\mathcal{C}_r$ in the input as a formal sum \eqref{counting_function}, where the coefficients $\alpha_i$ are written in binary form, the words $w_i$ are written in the alphabet $A_r$ and the summands are written as
pairs of a coefficient and a word. In the case of rational coefficients, we assume that the coefficients are written as fractions, where the numerator and denominator are written in binary form. (The  formalization of the input in the integer and rational cases can be found in \cite{ht-sbornik}.) This specification defines the size $\|f\|$ of the counting function $f$ provided as input.

\smallskip

We emphasize that the words $w_i$ in different summands may not be distinct: the transformation of a formal
sum into an equivalent one by combining coefficients at summands with equal words is not trivial.

\smallskip

Now we state the main result obtained in this paper.

\begin{theorem} \label{main}
    There exists an algorithm that takes as input two counting functions $f$~and~$g$ represented by linear combinations of elementary counting functions over the monoid $M_r$ and checks whether they are equivalent.
    Furthermore, this algorithm has time complexity 
    $O(r^3 n)$ for integer coefficients and $O(r^3n\log(rn))$ for rational coefficients, where $n=\|f\| + \|g\|$ is the input size.
\end{theorem}

The previous algorithm for solving this problem, which was constructed in \cite{ht-sbornik}, had two disadvantages: the description of the algorithm was quite complicated, and its complexity was proven to be linear only for $r\geq 3$. Even though the same algorithm works in the case of $r=2$, the complexity was only estimated as $O(n^2)$ for integer coefficients and $O(n^2\log n)$ for rational coefficients. This discrepancy has the following origin. For a counting function of size $n$ the algorithm makes $O(n)$ major steps, each of which has complexity $O(d)$, where $d$ is the current size of the stored data. Initially $d=n$, and after each major step a particular data compression procedure with compression ratio $p(r)=\max(\frac{2r}{r+2},\frac{2r^2}{3r-2},\frac{r^2}{2r-1})$ is executed. Since $p(r)\geq \frac97$ for all $r\geq 3$, in these cases the data size exponential decreasing along major steps, resulting in total complexity $O(n)$. However, if $r=2$, then $p(2)=1$, and in this case the data may potentially retain the same size along all the steps, which leads only to the estimation $O(n^2)$ (and $O(n^2\log n)$ respectively). 

\smallskip

Instead of constructing an algorithm for Theorem~\ref{main} directly, we will be looking at the algorithms that check whether a function is bounded.

\begin{theorem} \label{mainb}
    There exists an algorithm that takes as input a counting function $f$ represented by a linear combination of elementary counting functions over the monoid $M_r$ and checks whether $f$ is bounded.
    Furthermore, this algorithm has time complexity 
    $O(r^3\|f\|)$ for integer coefficients and time complexity $O(r^3\|f\|\log(r\|f\|))$ for rational coefficients, where $\|f\|$ is the input size.
\end{theorem}

The existence of an algorithm satisfying Theorem~\ref{main} easily follows from Theorem~\ref{mainb}.

\smallskip


Besides the described algebraic motivation, there is another topic closely related to counting functions on monoids. Recently, a number of results were obtained in \cite{abk,clmpw} on the boundedness problem for weighted automata (WA) and cost-register automata (CRA). These two classes of automata generalize probabilistic automata and carry additional computational information on transitions so that for any word $w$ from the automata alphabet $\Sigma$ a number $A(w)\geq 0$ can be computed. Given an automaton as input, the \emph{absolute boundedness problem} asks if there exists a constant $c>0$ such that for any $w\in \Sigma^*$ one has $|A(w)|\leq c$. The question whether a monoid counting function is bounded can be directly reduced to the absolute boundedness decision problem for linearly ambiguous WA and for the restricted model of copyless CRA. It was proven that the absolute boundedness problem is PSPACE-complete for polynomially ambiguous WA with integer coefficients (see \cite[Theorem~6.22]{abk}), and for copyless CRA it was proved to be decidable in polynomial time for non-negative coefficients \cite[Theorem~3.3]{clmpw}. However, restricting the coefficients to be non-negative does not allow a direct emulation of the monoid counting function anymore, since any counting function $\sum_{i=0}^k \alpha_i\rho_{w_i}$ with non-negative coefficients $\alpha_0,\ldots,a_k$ is clearly unbounded unless all these coefficients are zeros.


\smallskip

\section{Preliminary results}

It is not hard to see that for any fixed word $w\in M_r$ the functions 
\begin{equation}
\label{lr-functions}
\begin{aligned}
    l_w = \rho_w - \sum_{s \in A_r}\rho_{sw},\\
    r_w = \rho_w - \sum_{s \in A_r}\rho_{ws}
\end{aligned}
\end{equation}
are bounded, as for any argument word $x\in M_r$ one has $|l_w(x)|\leq 1$ and $|r_w(x)|\leq 1$. Indeed, any occurrence of the word $w$ in the word $x$ besides the prefix can be uniquely extended to an occurrence of the word $sw$ by the adjacent letter $s\in A_r$. Symmetrically, any occurrence of the word $w$ in the word $x$ besides the suffix is in one-to-one correspondence with its continuation $ws$ by some letter $s$. Since the functions $l_w$ and $r_w$ are bounded, they vanish in the quotient space $\mathcal{\widehat C}_r$.
This explains why $l_w$ and $r_w$ are called \emph{the left and right extension relation functions} respectively. 

\smallskip

A key structural insight, established in the work \cite{ggd}, shows that the functions \eqref{lr-functions} can be taken as the defining relations of the factor space over the bounded functions:

\begin{theorem} [Theorem~1.3 in \cite{ggd}, Theorem~1.4 in \cite{alt}]
    The subspace of bounded functions in $\mathcal C_r$ (i.\ e.\ ones equivalent to $0$ in the quotient space $\widehat{\mathcal C}_r$) is spanned by the left and right extension relation functions $l_w$ and $r_w$ taken for all words $w \in M_r$.
\end{theorem}

From this point on, for brevity we will also use counting functions to denote the equivalence classes they represent in $\widehat{\mathcal C}_r$. We will also say that \emph{applying the left (or right) extension relation} to an elementary counting function $\rho_{w}$ with $w = a_i v$ or $w = va_j$ means substituting $\rho_{w}$ with the linear combination
\[\rho_{v} - \sum_{s \in S \setminus \{a_i\}}\rho_{sv} \ \ \ \text{ (left extension)}\] 
and
\[\rho_{v} - \sum_{s \in S \setminus \{a_j\}}\rho_{vs} \ \ \ \text{ (right extension)} \]
respectively. This operation applied to a general counting function does not change its equivalence class because all extension relation functions are bounded. We will use this operation to rewrite the formal sums of type \eqref{counting_function} and obtain their basis representations.

\smallskip

The bases of the space $\widehat{\mathcal{C}}_r$ that consist solely of elementary counting functions can be parametrized using spanning trees of de Bruijn graphs (see \cite[Proposition 3.1]{ggd}), and the easiest such basis is described using the words satisfying the following simple condition:

\smallskip

\begin{defn}
    Let $\Wset_r$ be the set of all words from $M_r$, including the empty word $\varepsilon$, that do not start or end with the letter $a_1$.
\end{defn}


\begin{theorem} [Theorem~1.5 in \cite{ggd}, Theorem~1.5 in \cite{alt}]
\label{basis_thm}
    A basis of the space $\widehat{\mathcal{C}}_r$ is represented by the set  
    \[B_r = \{\rho_w \ \mid \ w \in \Wset_r \}.\]    
\end{theorem}

\medskip
The basis $B_r$ was used in \cite{alt} to find explicit formulas to decompose all elementary counting functions from $\mathcal{C}_r$ in the basis $B_r$. These formulas allow us to describe a natural and straightforward algorithm to check if the input function $f=\sum_{i=0}^k \alpha_i\rho_{w_i}$ is bounded: we represent all elementary counting functions from the input by linear combinations of the basis elements using the explicit formulas from \cite{alt}, and then combine the coefficients to see whether the result is a trivial combination (see formal details in \cite[Lemma~4.2]{ht-sbornik}). However, the application of this strategy with the basis $B_r$ to the input of size $n$ can produce a sum of counting functions that has length $\Theta(n^3)$, thus the time complexity of such an algorithm will be bounded below as $\Omega(n^3)$. The main technical tool of this paper is a new basis $\nbas_r$, which reduces the length of the produced sums to $O(n)$ and allows us to reach the desired linear time complexity. 

\smallskip

The following result describes the time complexity of an auxiliary formal procedure that sums up the coefficients in a general counting function with repetitions and serves as a tool to combine similar terms in the expressions of the form $\sum_{i=0}^k \alpha_i\rho_{w_i}$, where the functions $\rho_{w_i}$ may coincide.

\begin{lemma} [Lemma~4.2 in \cite{ht-sbornik}]
\label{lemm:normtime}
    There exists a procedure $N$ that takes a formal sum $f=\sum_{i=0}^k \alpha_i\rho_{w_i}$ as input and produces a formal sum obtained from $f$ by reducing the coefficients at summands with identical elementary counting functions. Furthermore, the time complexity of $N$
    is $O(rn)$ for integer coefficients and $O(rn\log(rn))$ for rational coefficients, where $n$ is the input size.
\end{lemma}

Suppose that the function $f$ is represented as a formal sum $f=\sum_{i=0}^k \alpha_i\rho_{w_i}$, where the elementary counting functions $\rho_{w_i}$ belong to some basis $T$ of the space $\widehat{\mathcal{C}}_r$, but may repeat. Applying the procedure $N$ to this sum produces a coordinate decomposition of $f$ in $T$, thus checking whether the function $f$ is bounded amounts to verifying whether the obtained decomposition is zero. The idea of the proposed algorithm is to obtain such a representation of the input function $f$, and then apply the procedure $N$ to it.

\section{Algorithm for $M_2$}
In this section, we describe an algorithm to check boundedness of a counting function
for the case of monoid $M_2$, and in the next section we outline its extension to the general~case.

\subsection{Basis decomposition} \label{sec:decomp}
First, for any elementary counting function $\rho_w$ we will find its explicit decomposition $D(\rho_w)$ in the basis $B_2$, which holds up to equivalence in $\widehat{\mathcal C}_2$. The decomposition formulas will be used extensively in our further considerations. 

\smallskip

First, we consider the general case of the word $w$ having form $w = a_1^kva_1^m$, where $v$ does not start or end with the letter $a_1$. Then we find formulas for the remaining words of the form $w=a_1^k$.

\begin{lemma} \label{lemm:repr}
    For any word $w = a_1^kva_1^m \in M_2$ with $v \in \nebas$
    it holds that $\rho_w \eqv D(\rho_w)$, where  
     \begin{equation}
        \label{key_decomposition}
        D(\rho_w)=\rho_v - \sum_{i = 0 \sbnd k-1}\rho_{a_2a_1^iv} -
        \sum_{j = 0 \sbnd m-1}\rho_{va_1^ja_2} +
        \sum_{\substack{i = 0 \sbnd k-1, \\ j = 0 \sbnd m-1}}\rho_{a_2a_1^iva_1^ja_2} \quad \text{for } k, m > 0 
        \end{equation}
    \begin{equation}
        \label{key_decomposition_k0}
        D(\rho_w)=\rho_v - \sum_{i = 0 \sbnd m-1}\rho_{va_1^ia_2} \quad \text{for } k = 0 \text{ and } m > 0,
    \end{equation}
    \begin{equation}
        \label{key_decomposition_m0}
        D(\rho_w)=\rho_v - \sum_{i = 0 \sbnd k-1}\rho_{a_2a_1^iv} \quad \text{for } k > 0 \text{ and } m = 0,
    \end{equation}
    and $D(\rho_w) = \rho_v = \rho_w$ for $k=m=0$.
\end{lemma}
\begin{proof}
    To establish the equality $\rho_w \eqv D(\rho_w)$ in the case \eqref{key_decomposition}
    we first iteratively apply the left extension relation
    $k$ times to the function $\rho_{a_1^kva_1^m}$ and obtain
    \begin{equation}
    \rho_w \eqv \rho_{va_1^m} - 
    \sum_{i = 0 \sbnd k-1}\rho_{a_2a_1^iva_1^m}.    
    \label{first_equality}
    \end{equation}
    Then, applying the right extension relation $m$ times to the function $\rho_{va_1^m}$,
    we get
    \[\rho_{va_1^m} \eqv \rho_v - 
    \sum_{j = 0 \sbnd m-1}\rho_{va_1^ja_2}.\]
    Finally, applying the right extension relation $i$ times to every element of the sum in \eqref{first_equality},
    we obtain
    \[\sum_{i = 0 \sbnd k-1}\rho_{a_2a_1^iva_1^m} \eqv
    \sum_{i = 0 \sbnd k-1}\Bigl (\rho_{a_2a_1^iv} - 
    \sum_{j = 0 \sbnd m-1}\rho_{a_2a_1^iva_1^ja_2} \Bigr )\eqv
    \sum_{i = 0 \sbnd k-1}\rho_{a_2a_1^iv} -
    \sum_{\substack{i = 0 \sbnd k-1 \\ j = 0 \sbnd m-1}}\rho_{a_2a_1^iva_1^ja_2},\]
    which yields the desired representation.

\smallskip

The same proof can be carried out in the cases \eqref{key_decomposition_k0} and \eqref{key_decomposition_m0} where $k$ or $m$ is equal to $0$.
First, note that \eqref{first_equality} from the proof of the Lemma~is exactly equivalent
to \eqref{key_decomposition_m0} when $m = 0$. Then, note that if instead of initially applying
the left extension relation $k$ times to $\rho_w$ we applied the right extension relation
$m$ times, we would have obtained
\[
\rho_w \eqv \rho_{a_1^kv} - 
\sum_{i = 0 \sbnd m-1}\rho_{a_1^kva_1^ia_2},
\]
which is in turn equivalent to \eqref{key_decomposition_k0} when $k = 0$.

The case of $k = m = 0$ where the decomposition $D$ is defined as an identity is obvious.
\end{proof}

\begin{lemma} \label{lemm:repr_degenerate}
    For every word $w = a_1^k \in M_2$ it holds that $\rho_w \eqv D(\rho_w)$, where 
    \begin{equation}
        \label{key_decomposition_degenerate}
        D(\rho_{w}) = \rho_{\varepsilon} - k\rho_{a_2} + \sum_{i=0 \sbnd k-2}(k - i - 1)\rho_{a_2a_1^ia_2}.
    \end{equation}
\end{lemma}
\begin{proof}
    Iteratively applying the left extension relation $k$ times to the function $\rho_{a_1^k}$, we get
    \[
    \rho_{w}\eqv \rho_{\varepsilon} - \sum_{i=0 \sbnd k-1}\rho_{a_2a_1^i}.
    \]
    Then, applying the right extension relation $i$ times to each function $\rho_{a_2a_1^i}$,
    we obtain 
    \[
    \rho_{a_2a_1^i} \eqv \rho_{a_2} - \sum_{j=0 \sbnd i-1}\rho_{a_2a_1^ja_2}.
    \]
    Finally, summing all of these representations together yields
    \begin{equation*}
    \rho_{w} \eqv \rho_{\varepsilon} - \sum_{i=0 \sbnd k-1}(\rho_{a_2} - \sum_{j=0 \sbnd i-1}\rho_{a_2a_1^ja_2}) \eqv
    \rho_{\varepsilon} - k\rho_{a_2} + \sum_{i=0 \sbnd k-2}(k - i - 1)\rho_{a_2a_1^ia_2}. \qedhere
    \end{equation*}
\end{proof}






It is easy to see that for any $w\in M_2$ the corresponding decomposition $D(\rho_w)$ from Lemma~\ref{lemm:repr} or Lemma~\ref{lemm:repr_degenerate} is a formal sum of elements in the basis $B_2$. Hence, the direct application of $D$ to the elementary counting functions in each summand of a formal sum representation of the input function allows us to obtain the desired basis representation.

\begin{corr}
    Let $f$ be a counting function in $\mathcal{C}_2$ represented by a linear combination
    $f=\sum_{i=0}^k \alpha_i\rho_{w_i}$. Then $f$ can be represented as the following combination
    of elements of $B_2$:
    \begin{equation}
      f \eqv \sum_{i=0}^k \alpha_i D(\rho_{w_i}).
      \label{d-formula}
    \end{equation}
\end{corr}

\subsection{A naive $O(n^3)$ algorithm} \label{sec:naive}
The decomposition obtained above allows us to construct a naive algorithm that checks whether $f$ is bounded. To do this, we use \eqref{d-formula} and replace each term $\alpha_i\rho_{w_i}$ of $f$ by the summands of $D(\rho_{w_i})$, each multiplied by $\alpha_i$. Then, we apply the procedure $N$ from Lemma~\ref{lemm:normtime} to the obtained formal sum. Since all the summands in the new representation contain elementary counting functions from $B_2$, the function $f$ is bounded if and only if the procedure $N$ yields a trivial linear combination.

\begin{remk} \normalfont \label{long_expr}
The time complexity of the naive algorithm is at least cubic, because the size of $D(\rho_{w})$
can be as large as $|w|^3$. For example, consider $w = a_1^ka_2a_1^k$. Then, due to \eqref{key_decomposition}
it holds that
\[\|D(\rho_w)\| \geq \Bigl \|\sum_{\substack{i = 0 \sbnd k-1 \\ j = 0 \sbnd k-1}}\rho_{a_2a_1^ia_2a_1^ja_2}\Bigr \| .\]
We see that, each summand with $i,j\geq k/2$ has size $\Theta(k)$, and there are $\Theta(k^2)$ such summands, from this one obtains that the total size of $D(\rho_w)$ is $\Theta(k^3) = \Theta(|w|^3)$. Therefore, the size of the basis representation is $\Omega(\|f\|^3)$, and the running time of the procedure $N$ can only be estimated as $O(\|f\|^3)$ for integers or $O(\|f\|^3\log \|f\|)$ for rationals.    
\end{remk}

\subsection{Shorthand basis}

As we saw in Remark \ref{long_expr}, the decomposition of an elementary counting function in the basis $B_2$ may significantly increase the data size, which in turn causes a summation of too many coefficients. To overcome this problem, we will construct another basis $\nbas_2$ which allows us to make short decompositions of the counting functions in the space $\mathcal{\widehat{C}}_2$. 
Clearly, the counting functions which are the sums from \eqref{key_decomposition} allow to compress the decomposition $D$, but the key point is that some of these sums can also constitute the most significant part of the basis $\nbas_2$.

\smallskip

Let us introduce the notation for the functions representing the sums from \eqref{key_decomposition}:
\begin{eqnarray*}
&\sigma_k: \nebas \to \mathbb{Z}, \quad &\sigma_k(v) = \sum_{i=0 \sbnd k-1}\rho_{a_2a_1^iv},\\
&\sigma^m: \nebas \to \mathbb{Z}, \quad &\sigma^m(v) = \sum_{j=0 \sbnd m-1}\rho_{va_1^ja_2},\\
&\sigma_k^m: \nebas \to \mathbb{Z}, \quad &\sigma_k^m(v) =
\sum_{\substack{i=0 \sbnd k-1, \\ j=0 \sbnd m-1}}\rho_{a_2a_1^iva_1^ja_2}.
\end{eqnarray*}
In this notation, the first case of Lemma~\ref{lemm:repr} can be reformulated as
\begin{equation}
\label{shorthand-formula}
\rho_{a_1^kva_1^m} \eqv \rho_v - \sigma_k(v) - \sigma^m(v) + \sigma_k^m(v).
\end{equation}
If we extend the definition of these notations to
\[\sigma_0(v) = \sigma^0(v) = \sigma_0^m(v) = \sigma_k^0(v) = 0,\]
then the equation \eqref{shorthand-formula} also holds for $k$ or $m$ being equal to $0$.

\medskip

Now, with the use of the counting functions $\sigma_k^m(v)$ we are ready to describe the basis $\nbas_2$, which will be used for the shorthand representations of any input counting function.

\begin{prop} \label{prop:main}
    Let \[\Sigma_2 = \{\sigma_{k}^{m}(v) \mid v \in \nebas, k, m > 0\}\]
    and let \[\abas_2 = \{\rho_{a_1^k} \mid k \geq 0\}.\]
    Then, $\nbas_2 = \Sigma_2 \cup \abas_2$ is a basis of $\widehat{\mathcal{C}}_2$. Furthermore,
    for an input function $f$ there exists a representation of $f$ in basis $\nbas_2$
    with size $O(\|f\|)$ that can be obtained in time $O(\|f\|)$.
\end{prop}

\medskip

The proof of the Proposition \ref{prop:main} is the main challenge of this paper. It is split into several lemmas and constitutes the Subsection~\ref{basis-proof}. Once such a representation is obtained, one can apply the procedure $N$ to it to check whether the input function is bounded, and due to the size of the representation being $O(\|f\|)$ the running time of $N$ will be $O(\|f\|)$ for integer coefficients and $O(\|f\|\log \|f\|)$ for rational coefficients, which satisfies
Theorem~\ref{mainb}.

\subsection{Basis proof}
\label{basis-proof}

Now, we prove that $\nbas_2$ is a basis for $\widehat{\mathcal{C}}_2$. First, we show that both sets $\abas_2$ and $\Sigma_2$ are linearly independent. Then, we will prove that $\spann(\abas_2)\cap \spann(\Sigma_2)=\emptyset$. Finally, we will demonstrate how any elementary counting function $\rho_w \in \mathcal{C}_2$ can be represented by an equivalent linear combination of functions from $\nbas_2$, proving that this set spans $\widehat{\mathcal{C}}_2$ and thus is a basis.

\begin{lemma} \label{lemm:a2indep}
    $\abas_2$ is a linearly independent set.
\end{lemma}
\begin{proof}
    It follows from Theorem~\ref{basis_thm} that the set $\abas_2' = \{a_2^k \mid k \geq 0\}$ is a subset
    of a basis for $\widehat{\mathcal{C}}_2$, thus this set is linearly independent. By symmetry, it follows that the set $\abas_2$ is also linearly independent, as it can be obtained from the set $\abas_2'$ by swapping letters
    $a_1$ and $a_2$ in the alphabet~$A_2$.
\end{proof}

Proving that the set $\Sigma_2$ is linearly independent is more involved, but also brings more insights into the structure of $\widehat{\mathcal{C}}_2$.

\begin{lemma} \label{lemm:s2indep}
    $\Sigma_2$ is a linearly independent set.
\end{lemma}

\smallskip

First, we prove that the functions $\sigma_k^m(v)$ for different words $v$ are linearly independent.

\begin{lemma} \label{lemm:sigma_subspaces}
    Suppose $v_1, v_2 \in \nebas$ are fixed words such that $v_1 \neq v_2$, then the sets $\{\sigma_k^m(v_1) \mid k, m > 0\}$ and $\{\sigma_k^m(v_2) \mid k, m > 0\}$ span in $\widehat{C}_2$ two subspaces which have a trivial intersection.
\end{lemma}
\begin{proof}
    Recall from the definition of $\sigma$ that
    \[\sigma_k^m(v) =
    \sum_{\substack{i=0 \sbnd k-1 \\ j=0 \sbnd m-1}}\rho_{a_2a_1^iva_1^ja_2},\]
    therefore for any $k, m > 0$ it holds that
    \begin{align*}
     &\sigma_k^m(v_1) \in V_1 = \spann(\{\rho_{a_2a_1^iv_1a_2^ja_2} \mid i, j \geq 0\}), \\
     &\sigma_k^m(v_2) \in V_2 = \spann(\{\rho_{a_2a_1^iv_2a_2^ja_2} \mid i, j \geq 0\}).   
    \end{align*}
    Now, since $v_1 \neq v_2$, it follows that the subspaces $V_1$ and $V_2$ are spanned by disjoint subsets of the
    basis $B_2$, thus they intersect trivially, therefore
    $\{\sigma_k^m(v_1) \mid k, m > 0\}$ and $\{\sigma_k^m(v_2) \mid k, m > 0\}$ also span subspaces which intersect trivially.
\end{proof}

Up to now we have established that $\Sigma_2$ is a direct sum of the sets $\{\sigma_k^m(v) \mid k, m > 0\}$, where $v$ runs over all words in $\nebas$. It remains to show that for any fixed word $v$ each individual set $\{\sigma_k^m(v) \mid k, m > 0\}$ is linearly independent.

\smallskip

\begin{lemma} \label{lemm:sigma_fixed_word}
    Fix a word $v \in \nebas$ and let $\Sigma_2^v$ be the set
    $\{\sigma_k^m(v) \mid k, m > 0\}$. Then, $\Sigma_2^v$~is~a linearly independent set.
\end{lemma}

Let us fix a word $v \in \nebas$ and denote by $\Sigma_2^v$ the set
$\{\sigma_k^m(v) \mid k, m > 0\}$. Remarkably, the elements of $\Sigma_2^v$ can naturally be 
graphically depicted as follows. Let us consider a $2$-dimensional 
weighted multiset of rectangles parallel to the coordinates, which we call a \emph{histogram}.

\begin{defn} \label{def:hist_elem}
    Let $\rho_w$ be the elementary counting function for the word $w = a_2a_1^iva_1^ja_2$, where $i,j\geq 0$ and let $\alpha \in \mathbb R$. The histogram representation of the counting function $\alpha\rho_w$ is a $1\times 1$ square with its bottom left corner in point $(i, j)$ and the weight equal to $\alpha$.
\end{defn}

This definition leads to the histogram representation of the counting function $\sigma_k^m(v)$.

\begin{defn} \label{def:hist_sigma}
    The histogram representation of the function $\alpha\sigma_k^m(v)$ is a sum of
    the histogram representations of $\{\alpha\rho_{a_2a_1^iva_1^ja_2} \mid 0 \leq i < k, 0 \leq j < m\}$,
    or an origin-based rectangle with corners in $(0, 0)$, $(k, 0)$, $(0, m)$ and $(k, m)$ and the weight equal to $\alpha$.
\end{defn}

The unit square and origin-based rectangle histograms can be seen in Figure \ref{fig:hist}.

\tikzset{every picture/.style={line width=0.75pt}} 

\medskip

\begin{figure}[H]
    \centering

\begin{tikzpicture}[x=0.75pt,y=0.75pt,yscale=-1,xscale=1]

\draw  (64,240.29) -- (314,240.29)(90.43,11) -- (90.43,263) (307,235.29) -- (314,240.29) -- (307,245.29) (85.43,18) -- (90.43,11) -- (95.43,18)  ;
\draw   (90.43,60) -- (280.2,60) -- (280.2,240.29) -- (90.43,240.29) -- cycle ;
\draw   (260.43,60) -- (280.2,60) -- (280.2,80) -- (260.43,80) -- cycle ;
\draw    (280.2,240.29) -- (280.2,250) ;
\draw    (80.2,60) -- (90.43,60) ;

\draw (73,13.4) node [anchor=north west][inner sep=0.75pt]    {$y$};
\draw (294,245.4) node [anchor=north west][inner sep=0.75pt]    {$x$};
\draw (62,49.8) node [anchor=north west][inner sep=0.75pt]    {$k$};
\draw (276,251.4) node [anchor=north west][inner sep=0.75pt]    {$m$};
\draw (174,140.4) node [anchor=north west][inner sep=0.75pt]    {$\sigma_k^m$};
\draw (297.83,34) node [anchor=south] [inner sep=0.75pt]    {$\rho_{a_2a_1^{k-1}va_1^{m-1}a_2}$};
\draw (157.22,64.4) node [anchor=north] [inner sep=0.75pt]    {};
\draw (260.43,63.4) node [anchor=north] [inner sep=0.75pt]    {};
\draw (277.65,165.2) node [anchor=east] [inner sep=0.75pt]    {};
\draw    (287.77,38.4) -- (270.71,58.74) ;
\draw [shift={(269.43,60.27)}, rotate = 309.98] [color={rgb, 255:red, 0; green, 0; blue, 0 }  ][line width=0.75]    (10.93,-3.29) .. controls (6.95,-1.4) and (3.31,-0.3) .. (0,0) .. controls (3.31,0.3) and (6.95,1.4) .. (10.93,3.29)   ;

\end{tikzpicture}
    
    \caption{Histogram example for the function $\sigma_k^m$.}
    \label{fig:hist}
\end{figure}

Finally, consider an arbitrary linear combination $f_v = \sum_{i=1}^{n}\alpha_i\sigma_{k_i}^{m_i}(v)$
of elements of $\Sigma_2^v$. Without loss of generality, we assume that each pair $(k_i, m_i)$ is distinct.
The histogram representation of $f_v$ is, naturally, the sum of histograms for every summand $\alpha_i\sigma_{k_i}^{m_i}(v)$.
The histogram for each summand histogram is a weighted rectangle by Definition \ref{def:hist_sigma}.
Suppose that this linear combination is trivial, or the weight of the sum histogram is zero at all points.
We now prove that the weight of every summand rectangle histogram must also be zero.

\smallskip

Let us denote by $H(x, y)$ the sum weight at the square with coordinates $(x, y)$, and
let us denote by $h(x, y)$ the weight of the summand rectangle with the extent $(x, y)$.
In terms of functions, $H(x, y)$ is the projection of $f_v$ onto $\rho_{a_2a_1^xva_1^ya_2}$, and
$h(x, y)$ is the coefficient of $\sigma_x^y(v)$ in the linear combination representation of $f_v$.
Note that from the histogram definition it follows that
\[ H(x, y) = \sum_{x' \geq x,\, y' \geq y}h(x', y'). \]

\begin{lemma} \label{lemm:hist}
    $H \equiv 0$ if and only if $h \equiv 0$.
\end{lemma}
\begin{proof}
    If $h \equiv 0$, then clearly so is $H$. Now, assume that $H \equiv 0$.
    Then we prove that $h \equiv 0$ by descending induction on $x$ and then $y$.

    As the base case we take the maximum value of $x_0$ over all summand rectangles, and the maximum $y_{00}$ value for $x = x_0$. According to maximality, the only rectangles containing $(x_0, y_{00})$ are those which have exactly those coordinates,
    thus $h(x_0, y_{00}) = H(x_0, y_{00}) = 0$, and the base case is proven.

    Now let us consider some $(x, y)$ and let us assume that by the induction hypothesis
    for all $(x', y')$ such that $x' > x$ or $x' = x$ and $y' > y$ it holds that
    $h(x', y') = 0$. Then,
    \[ 0 = H(x, y) = \sum_{x' \geq x,\, y' \geq y}h(x', y') = h(x, y) +  0 = h(x, y), \]
    thus $h(x, y) = 0$ and the induction step is proven.
\end{proof}

Now, from this we can immediately derive the desired linear independences.

\begin{proof}[Proof of Lemma~\ref{lemm:sigma_fixed_word}]
    From Lemma~\ref{lemm:hist} it follows that the histogram corresponding
    to any linear combination of functions
    $\sigma_{k}^{m}(v) \in \Sigma_2^v$
    is trivial (i.e., $H \equiv 0$) if and only if all coefficients are equal to zero (i.e., $h \equiv 0$),
    thus $\Sigma_2^v$ is indeed a linearly independent set.
\end{proof}

\begin{proof}[Proof of Lemma~\ref{lemm:s2indep}]
     The set $\Sigma_2^v$ is linearly independent for every word $v \in \nebas$
     by Lemma~\ref{lemm:sigma_fixed_word},
     and the subspaces spanned by $\Sigma_2^v$ intersect trivially by Lemma~
     \ref{lemm:sigma_subspaces}, therefore the set
     $\Sigma_2 = \bigcup_{v \in \nebas}\Sigma_2^v$ is linearly independent.
\end{proof}

\medskip

Now that we have proven that $\abas_2$ and $\Sigma_2$ are linearly independent sets, we will show that they span linearly independent subspaces (or that their union $\nbas_2$ is linearly independent).

\begin{lemma} \label{lemm:b2indep}
    The span of $\abas_2$ lies in the span of the set
    \[B_2|_{\abas_2} = \{\rho_{a_2a_1^ka_2}  \mid  k \geq 0\} \cup \{\rho_{a_2}\} \cup \{\rho_{\varepsilon}\} \subset B_2\]
    and the span of $\Sigma_2$ lies in the span of the set
    \[B_2|_{\Sigma_2} = \{\rho_{a_2a_1^kva_1^ma_2}  \mid  v \in \nebas, k, m \geq 0\} \subset B_2.\]
\end{lemma}
\begin{proof}
    The statement of the Lemma~for $\abas_2$ follows immediately from equation \eqref{key_decomposition_degenerate}
    in the definition of $D$ for the words $a_1^k$. The statement of the Lemma~for $\Sigma_2$ follows from the definition of $\sigma_k^m(v)$.
\end{proof}

Now note that $B_2|_{\abas_2}$ and $B_2|_{\Sigma_2}$ are disjoint subsets of the basis $B_2$, thus their spans
are linearly independent and the spans of $\abas_2$ and $\Sigma_2$ are also linearly independent.

\medskip

Finally, we prove that the set $\nbas_2 = \abas_2 \cup \Sigma_2$ spans the whole space $\widehat{\mathcal{C}}_2$.

\begin{lemma} \label{lemm:b2span}
    The elementary function $\rho_w$ for each word $w \in M_2$ is representable
    as a linear combination of elements of $\nbas_2$.
    Furthermore,
    \begin{enumerate}[label=\arabic*.]
        \item If $w = \varepsilon$ or $w = a_1^k$, then $\rho_w \in \abas_2$.
        \item If $w = a_2$, then $\rho_w \eqv \rho_{\varepsilon} - \rho_{a_1}$.
        \item If $w = a_2a_1^k$ or $w = a_1^ka_2$ with $k > 0$, then $\rho_w \eqv \rho_{a_1^k} - \rho_{a_1^{k+1}}$.
        \item If $w = a_2a_1^ka_2$ with $k \geq 0$, then $\rho_w \eqv \rho_{a_1^k} - 2\rho_{a_1^{k+1}} + \rho_{a_1^{k+2}}$.
        \item If $w = a_1^ka_2a_1^m$ with $k,m > 0$, then
        $\rho_w = -\rho_{\varepsilon} + \rho_{a_1} + \rho_{a_1^k} + \rho_{a_1^m}
        - \rho_{a_1^{k+1}} - \rho_{a_1^{m+1}} + \sigma_k^m(a_2).$
        \item If $w = a_1^ka_2a_1^na_2a_1^m$ with $k,m,n \geq 0$, then
        \[
        \rho_w = \rho_{a_1^n} - 2\rho_{a_1^{n+1}} + \rho_{a_1^{n+2}} - \sigma_k^{n + 1}(a_2)
        + \sigma_k^n(a_2) - \sigma_{n + 1}^m(a_2) + \sigma_n^m(a_2)
        + \sigma_k^m(a_2a_1^na_2).
        \]
        \item Otherwise, $w$ can be written in the form $w = a_1^ka_2a_1^{k'}va_1^{m'}a_2a_1^m$ with $k, m, k', m' \geq 0$ and
        $v \in \nebas$, and
        \[
        \rho_w =
        \sigma_{k' + 1}^{m' + 1}(v) - \sigma_{k'}^{m' + 1}(v) - \sigma_{k' + 1}^{m'}(v) + \sigma_{k'}^{m'}(v)
        -\]\[
        - \sigma_k^{m' + 1}(a_2a_1^{k'}v)
        + \sigma_k^{m'}(a_2a_1^{k'}v)
        - \sigma_{k' + 1}^m(va_1^{m'}a_2)
        + \sigma_{k'}^m(va_1^{m'}a_2) +
        \]\[
        + \sigma_k^m(a_2a_1^{k'}va_1^{m'}a_2).
        \]
    \end{enumerate}
\end{lemma}
\begin{proof}
    Case 1 trivially follows from the definition of $\abas_2$.
    Case 2 is proven by applying the left extension relation to $\rho_{a_2}$.

    \smallskip
    
    Case 3 is proven by applying the left extension relation to $\rho_{a_2a_1^k}$ and the right extension
    relation to $\rho_{a_1^ka_2}$ respectively.
    
    \smallskip
    
    Case 4 is proven by applying the right extension relation to $\rho_{a_2a_1^ka_2}$ and then Case~3
    to the two summands as follows:
    \[\rho_{a_2a_1^ka_2} \eqv \rho_{a_2a_1^k} - \rho_{a_2a_1^{k+1}} \eqv 
    (\rho_{a_1^k} - \rho_{a_1^{k+1}}) - (\rho_{a_1^{k+1}} - \rho_{a_1^{k+2}}) = 
    \rho_{a_1^k} - 2\rho_{a_1^{k+1}} + \rho_{a_1^{k+2}}.\]
    
    \medskip

    Now, Cases 5-7 are the three different cases of how the representation \eqref{shorthand-formula}
    is further decomposed into a combination of $\nbas_2$ elements.
    
    \smallskip

    Case 5 is the simpler of the three. In this case the equivalence \eqref{shorthand-formula} takes the form
    
    \[\rho_{a_1^ka_2a_1^m} \eqv \rho_{a_2} - \sigma_k(a_2) - \sigma^m(a_2) + \sigma_k^m(a_2) =
    \rho_{a_2} - \sum_{i = 0 \sbnd k-1}\rho_{a_2a_1^ia_2} - \sum_{i = 0 \sbnd m-1}\rho_{a_2a_1^ia_2} +
     \sigma_k^m(a_2).\]

    We substitute the first summand using Case 2 and we substitute all the elements in the expanded
    sums using Case 4, obtaining
    \[(\rho_{\varepsilon} - \rho_{a_1}) - \sum_{i = 0 \sbnd k-1}(\rho_{a_1^i} - 2\rho_{a_1^{i + 1}} +
    \rho_{a_1^{i + 2}}) - 
    \sum_{i = 0 \sbnd m-1}(\rho_{a_1^i} - 2\rho_{a_1^{i + 1}} +
    \rho_{a_1^{i + 2}}) + \sigma_k^m(a_2) =\]\[=
    (\rho_{\varepsilon} - \rho_{a_1}) -
    (\rho_{\varepsilon} - \rho_{a_1} - \rho_{a_1^k} + \rho_{a_1^{k+1}}) -
    (\rho_{\varepsilon} - \rho_{a_1} - \rho_{a_1^m} + \rho_{a_1^{m+1}})  + \sigma_k^m(a_2) =\]\[ =
    -\rho_{\varepsilon} + \rho_{a_1} + \rho_{a_1^k} + \rho_{a_1^m}
        - \rho_{a_1^{k+1}} - \rho_{a_1^{m+1}} + \sigma_k^m(a_2)\]

    \smallskip

    Now, in Cases 6 and 7 we run into the problem of representing $\sigma_k(v)$ and $\sigma^m(v)$, as they are not part of the identified basis $\nbas_2$. The key insight which will allow us to do that lies in the structure of the word $v$.
    Let us begin with the simpler Case~6. Here, the equivalence takes the form 
    \[\rho_{a_1^ka_2a_1^na_2a_1^m} \eqv
    \rho_{a_2a_1^na_2} - \sigma_k(a_2a_1^na_2) - \sigma^m(a_2a_1^na_2) + \sigma_k^m(a_2a_1^na_2) \eqv\]\[    \eqv (\rho_{a_1^n} - 2\rho_{a_1^{n+1}} + \rho_{a_1^{n+2}}) - \sigma_k(a_2a_1^na_2) - \sigma^m(a_2a_1^na_2) + \sigma_k^m(a_2a_1^na_2),
    \]
    where we have used Case 4 to further decompose the first summand. Now note that the 
    argument word $a_2a_1^na_2$ is both of the form
    $v'a_1^na_2$ (with $v' = a_2$) and of the form $a_2a_1^nv''$ (with $v'' = a_2$ as well).
    This fact allows us to represent $\sigma_k(a_2a_1^na_2)$ as a linear
    combination of elements of $\Sigma_2^{v'}$ in the following way:
    \begin{equation}\label{col-dec}
    \sigma_k(a_2a_1^na_2) = \sigma_k(v'a_1^na_2) = 
    \end{equation}
    \[
    = \sum_{i=0 \sbnd k-1}\rho_{a_2a_1^iv'a_1^na_2} =
    \sum_{\substack{i=0 \sbnd k-1, \\ j=0 \sbnd n}}\rho_{a_2a_1^iv'a_1^ja_2} - 
    \sum_{\substack{i=0 \sbnd k-1, \\ j=0 \sbnd n - 1}}\rho_{a_2a_1^iv'a_1^ja_2} = \]
    \[ =
    \sigma_k^{n+1}(v') - \sigma_k^n(v') = \sigma_k^{n+1}(a_2) - \sigma_k^n(a_2).
    \]
    The function $\sigma_k(a_2a_1^na_2)$ is represented as a linear
    combination of elements of $\Sigma_2^{v''}$ in an analogous fashion:
    \begin{equation}\label{row-dec}
    \sigma^m(a_2a_1^na_2) = \sigma^m(a_2a_1^nv'') = 
    \end{equation}
    \[ = \sum_{i=0 \sbnd m-1}\rho_{a_2a_1^nv''a_1^ia_2} =
    \sum_{\substack{j=0 \sbnd n, \\ i=0 \sbnd m-1}}\rho_{a_2a_1^jv''a_1^ia_2} - 
    \sum_{\substack{j=0 \sbnd n-1, \\ i=0 \sbnd m-1}}\rho_{a_2a_1^jv''a_1^ia_2} = \]
    \[ =
    \sigma_{n+1}^m(v'') - \sigma_n^m(v'') = \sigma_{n+1}^m(a_2) - \sigma_n^m(a_2).
    \]

    Substituting equations \eqref{col-dec} and \eqref{row-dec} with the equivalence \eqref{shorthand-formula} yields the desired representation for Case 6:
    \[\rho_{a_1^ka_2a_1^na_2a_1^m} \eqv (\rho_{a_1^n} - 2\rho_{a_1^{n+1}} + \rho_{a_1^{n+2}}) - \sigma_k(a_2a_1^na_2) - \sigma^m(a_2a_1^na_2) + \sigma_k^m(a_2a_1^na_2) =\]\[ =
    \rho_{a_1^n} - 2\rho_{a_1^{n+1}} + \rho_{a_1^{n+2}} - \sigma_k^{n + 1}(a_2)
        + \sigma_k^n(a_2) - \sigma_{n + 1}^m(a_2) + \sigma_n^m(a_2)
        + \sigma_k^m(a_2a_1^na_2).\]

    Finally, in Case 7 the same idea is generalized even further. Here, the equivalence \eqref{shorthand-formula} takes the form
    \[\rho_{a_1^ka_2a_1^{k'}va_1^{m'}a_2a_1^m} \eqv
    \rho_{a_2a_1^{k'}va_1^{m'}a_2} - \sigma_k(a_2a_1^{k'}va_1^{m'}a_2) - \sigma^m(a_2a_1^{k'}va_1^{m'}a_2) + \sigma_k^m(a_2a_1^{k'}va_1^{m'}a_2).\]
    Here we can use the formulas \eqref{col-dec} and \eqref{row-dec} to decompose the second and third summands by taking $v' = a_2a_1^{k'}v$
    and $v'' = va_1^{m'}a_2$:
    \[\sigma_k(a_2a_1^{k'}va_1^{m'}a_2) = \sigma_k^{m' + 1}(a_2a_1^{k'}v) - \sigma_k^{m'}(a_2a_1^{k'}v),\]
    \[\sigma^m(a_2a_1^{k'}va_1^{m'}a_2) = \sigma_{k' + 1}^m(va_1^{m'}a_2) - \sigma_{k'}^m(va_1^{m'}a_2).\]
    The only summand left is $\rho_{a_2a_1^{k'}va_1^{m'}a_2}$. Although this elementary function belongs
    to the previous basis $B_2$, in the new basis $\nbas_2$ it has to be further represented. However,
    we can use a similar technique twice to obtain such a representation:
    \begin{equation}\label{point-dec}
    \rho_{a_2a_1^{k'}va_1^{m'}a_2} = \sum_{i=0 \sbnd k'}\rho_{a_2a_1^iva_1^{m'}a_2} -
    \sum_{i=0 \sbnd k' - 1}\rho_{a_2a_1^iva_1^{m'}a_2} =
    \end{equation}
    \[ = \sigma_{k' + 1}(va_1^{m'}a_2) - 
    \sigma_{k'}(va_1^{m'}a_2) =(\sigma_{k' + 1}^{m' + 1}(v) - \sigma_{k' + 1}^{m'}(v)) - (\sigma_{k'}^{m' + 1}(v) - \sigma_{k'}^{m'}(v)) =\]\[= \sigma_{k' + 1}^{m' + 1}(v) - \sigma_{k'}^{m' + 1}(v) - \sigma_{k' + 1}^{m'}(v) + \sigma_{k'}^{m'}(v).\]

    Substituting the first three summands in the equivalence using equations \eqref{col-dec}, \eqref{row-dec} and \eqref{point-dec} leaves us with the desired representation of 
    $\rho_{a_1^ka_2a_1^na_2a_1^m}$ from Case 7.
\end{proof}

\smallskip

Finally, we are ready to conclude the proof of Proposition \ref{prop:main}.

\begin{proof}[Proof of Proposition \ref{prop:main}]
It follows from Lemmas \ref{lemm:a2indep}, \ref{lemm:s2indep} and \ref{lemm:b2indep} that
$\nbas_2 = \abas_2 \cup \Sigma_2$ is a linearly independent set, and from Lemma~\ref{lemm:b2span} it follows that
$\nbas_2$ is a spanning set of $\widehat{\mathcal{C}}_2$, therefore $\nbas_2$ is indeed a basis.
Furthermore, the formulation of Lemma~\ref{lemm:b2span} demonstrates how any elementary counting
function can be represented as a linear combination of a constant number of elements of $\nbas_2$ (actually, no more than $9$) with coefficients $+1$, $-1$ and $-2$. If we apply this decomposition to each summand in a formal sum representing a function $f$, it is clear that the size of the decomposed sum exceeds $\|f\|$ by
no more than a constant multiplicative factor, concluding the proof.
\end{proof}

\subsection{The main algorithm} \label{sec:algo}
We have identified a basis that allows us to obtain a shorthand representation of
any input function $f$, so now we can describe an algorithm that satisfies Theorem~\ref{mainb}.

\smallskip

Let $f$ be represented by a linear combination $\sum_{i=1}^{n} \alpha_i\rho_{w_i}$, where each summand is encoded as a word-coefficient pair.

We encode each summand consisting of a coefficient $\alpha$ and an elementary counting function
$\rho_w$ as a pair $(\alpha, w)$, and we encode the function $\sigma_k^m(v)$ with a coefficient $\alpha$ as a compressed tuple $\langle \alpha, k, m, v\rangle$. Since we have proven that $\nbas_2$ is a basis, we can use a mixed list of pairs $(\alpha, a_1^n)$ and tuples as input to the procedure $N$, where tuples and pairs can be processed separately, as the functions they represent are located in disjoint subspaces. Furthermore, this compressed tuple form is what allows us to avoid the size increase from double summations \--- the size of the tuple $\langle \alpha, k, m, v\rangle$ is proportional to just $\|\alpha\| + k + m + |v|$.

\medskip

We replace each pair $(\alpha_i, w_i)$ from the input function by other pairs and tuples in one of the following ways based on the word $w_i$, mirroring
the cases from Lemma~\ref{lemm:b2span}:
\begin{enumerate}[label=\arabic*.]
    \item If $w_i = \varepsilon$ or $w_i = a_1^k$, we leave the pair as is. 
    \smallskip
    
    \item Else if $w_i = a_2$, we replace by $(\alpha_i, \varepsilon)$, $(-\alpha_i, a_1)$.
   \smallskip
    
    \item Else if $w_i = a_2a_1^k$ or $w_i = a_1^ka_2$ with $k > 0$,
    we replace by $(\alpha_i, a_1^k)$, $(-\alpha_i, a_1^{k+1})$.
    \smallskip
        
    \item Else if $w_i = a_2a_1^ka_2$ with $k \geq 0$,
    we replace by $(\alpha_i, a_1^k)$, $(-2\alpha_i, a_1^{k+1})$, $(\alpha_i, a_1^{k+2})$.
    \smallskip
    
    \item Else if $w_i = a_1^ka_2a_1^m$ with $k,m > 0$,
    we replace by
    $(-\alpha_i, \varepsilon)$, $(\alpha_i, a_1)$, $(\alpha_i, a_1^k)$, $(\alpha_i, a_1^m)$,
    $(-\alpha_i, a_1^{k+1})$, $(-\alpha_i, a_1^{m+1})$,
    $\langle\alpha_i, k, m, a_2\rangle$.
    \smallskip
    
    \item Else if $w_i = a_1^ka_2a_1^na_2a_1^m$ with $k,m,n \geq 0$,
    we replace by
    $(\alpha_i, a_1^n)$, $(-2\alpha_i, a_1^{n+1})$, $(\alpha_i, a_1^{n+2})$,
    $\langle -\alpha_i, k, n + 1, a_2\rangle$, $\langle \alpha_i, k, n, a_2\rangle$,
    $\langle -\alpha_i, n + 1, m, a_2\rangle$, $\langle \alpha_i, n, m, a_2\rangle$,\\
    $\langle\alpha_i, k, m, a_2a_1^na_2\rangle$.
    \smallskip
    
    \item Otherwise, $w_i = a_1^ka_2a_1^{k'}va_1^{m'}a_2a_1^m$
    with $k,m,k',m' \geq 0$ and $v \in \nebas$,
    and we replace by
    $\langle\alpha_i, k' + 1, m' + 1, v\rangle$,
    $\langle-\alpha_i, k', m' + 1, v\rangle$,
    $\langle-\alpha_i, k' + 1, m', v\rangle$,
    $\langle\alpha_i, k', m', v\rangle$,
    $\langle -\alpha_i, k, m' + 1, a_2a_1^{k'}v\rangle$, $\langle \alpha_i, k, m', a_2a_1^{k'}v\rangle$,
    $\langle -\alpha_i, k' + 1, m, va_1^{m'}a_2\rangle$, $\langle \alpha_i, k', m, va_1^{m'}a_2\rangle$,
    $\langle\alpha_i, k, m, a_2a_1^{k'}va_1^{m'}a_2\rangle$.
\end{enumerate}

Finally, we apply the procedure $N$ to the obtained list. If the result is a trivial sum, then we conclude that $f$ is bounded; otherwise, we conclude that it is not.
Scanning the input and matching each word $w_i$ to one of the 7 cases above clearly takes $O(\|f\|)$ time. 
Now note that the size of each new pair or tuple is proportional to the size of the pair that has been substituted, and each pair is substituted by no more than $9$ new elements (each of which can be a tuple or a pair).
Thus, the size of the input to the invocation of $N$ is $O(\|f\|)$. Hence, according to Lemma~\ref{lemm:normtime}, the running time of the procedure $N$ is $O(\|f\|)$ for integer coefficients and $O(\|f\|\log \|f\|)$ for rational coefficients, thereby establishing Theorem~\ref{mainb} for the case $M_2$.

\section{Generalization for $M_r$, $r \geq 2$}
For the general case of the monoid $M_r$ with an arbitrary $r \geq 2$, a similar method is also applicable.
Here we omit detailed proofs, since the core ideas in the general case are identical while the equations and exact formulations are more cumbersome.

\medskip

First, let us extend the basis decomposition from Lemma~\ref{lemm:repr} to the general case.

\begin{lemma} \label{lemm:reprn}
    For any word $w = a_1^kva_1^m \in M_r$ with $v \in \nebasr$
    it holds that $\rho_w \eqv D_r(\rho_w)$, where  
     \begin{equation}
        \label{key_decomposition_n}
        D_r(\rho_w)=\rho_v - \sum_{\substack{i = 0 \sbnd k-1, \\ x \in A_r \setminus \{a_1\}}}\rho_{xa_1^iv} -
        \sum_{\substack{j = 0 \sbnd m-1, \\ x \in A_r \setminus \{a_1\}}}\rho_{va_1^jx} +
        \sum_{\substack{i = 0 \sbnd k-1, \\ j = 0 \sbnd m-1, \\ x_1, x_2 \in A_r \setminus \{a_1\}}}\rho_{x_1a_1^iva_1^jx_2} \quad \text{for } k, m > 0 
        \end{equation}
    \begin{equation}
        \label{key_decomposition_k0_n}
        D_r(\rho_w)=\rho_v - \sum_{\substack{i = 0 \sbnd m-1, \\ x \in A_r \setminus \{a_1\}}}\rho_{va_1^ix} \quad \text{for } k = 0 \text{ and } m > 0,
    \end{equation}
    \begin{equation}
        \label{key_decomposition_m0_n}
        D_r(\rho_w)=\rho_v - \sum_{\substack{i = 0 \sbnd k-1, \\ x \in A_r \setminus \{a_1\}}}\rho_{xa_1^iv} \quad \text{for } k > 0 \text{ and } m = 0,
    \end{equation}
    and $D_r(\rho_w) = \rho_v = \rho_w$ for $k=m=0$.
\end{lemma}
\begin{proof}
    The proof is identical to the proof of Lemma~\ref{lemm:repr}.
\end{proof}

Note that this representation is exactly the same as the one for the case of $M_2$, but wherever we had the letter $a_2$ delimiting the word, in the general case one has a summation over all the letters of
$A_r \setminus \{a_1\}$. As in the case of $M_2$, we also extend the definition of $D_r$ to all
the other words in $M_r$.

\begin{lemma} \label{lemm:repr_degeneraten}
    For every word $w = a_1^k \in M_r$ it holds that $\rho_w \eqv D_r(\rho_w)$, where 
    \begin{equation}
        \label{key_decomposition_degenerate_n}
        D_r(\rho_{w}) = \rho_{\varepsilon} - \sum_{x \in A_r \setminus \{a_1\}}k\rho_{x} + \sum_{\substack{i=0 \sbnd k-2, \\ x_1, x_2 \in A_r \setminus \{a_1\}}}(k - i - 1)\rho_{x_1a_1^ix_2}.
    \end{equation}
\end{lemma}
\begin{proof}
    The proof is identical to the proof of Lemma~\ref{lemm:repr_degenerate}.
\end{proof}

\smallskip

The decomposition \eqref{key_decomposition_degenerate_n} allows us to extend the naive algorithm from Section \ref{sec:naive} to the general case of $M_r$. The only difference is in the running time: in the decomposition $D_r$ each summand of the form $\rho_{x_1a_1^iva_1^jx_2}$ is included $(r-1)^2$ times, which gives
the estimated size $O(r^2n^3)$. Therefore, the running time of the procedure
$N$ when applied to this representation can only be estimated as
$O(r^3n^3)$ for integer coefficients and $O(r^3n^3\log(rn))$ for rational coefficients.

\medskip

In the following, we will show that a shorthand basis exists in the general case. The notation for the sums from \eqref{key_decomposition_n} is defined quite similarly to Section \ref{sec:decomp}, but now we need to introduce new letter parameters $x$, $x_1$ and $x_2$.

\begin{eqnarray*}
&\sigma_k: (\nebasr) \times (A_r \setminus \{a_1\}) \to \mathbb{Z}, \quad &\sigma_k(v, x) = \sum_{i = 0 \sbnd k-1}\rho_{xa_1^iv},\\
&\sigma^m: (\nebasr) \times (A_r \setminus \{a_1\}) \to \mathbb{Z}, \quad &\sigma^m(v, x) = \sum_{j = 0 \sbnd m-1}\rho_{va_1^jx},\\
&\sigma_k^m: (\nebasr) \times (A_r \setminus \{a_1\})^2 \to \mathbb{Z}, \quad &\sigma_k^m(v, x_1, x_2) =
\sum_{\substack{i = 0 \sbnd k-1, \\ j = 0 \sbnd m-1}}\rho_{x_1a_1^iva_1^jx_2}.
\end{eqnarray*}

Note that in the case of $M_2$ the set $A_r \setminus \{a_1\}$ degenerates into $\{a_2\}$, thus
the letters $x, x_1$ and $x_2$ do not appear as parameters for $\sigma_k, \sigma^m$ and $\sigma_k^m$
as they can only have the one value.

\smallskip

Now, in this notation the first case of Lemma~\ref{lemm:reprn} can be reformulated as

\[\rho_{a_1^kva_1^m} \eqv \rho_v -
\sum_{x \in A_r\setminus\{a_1\}}\sigma_k(v, x) -
\sum_{x \in A_r\setminus\{a_1\}}\sigma^m(v, x) +
\sum_{x_1, x_2 \in A_r\setminus\{a_1\}}\sigma_k^m(v, x_1, x_2),\]
where there are only $O(r^2)$ shorthand summands. Now, we will describe the basis of 
$\widehat{\mathcal{C}}_r$ using the aforementioned sums.

\begin{prop} \label{prop:nmain}
    Let \[\Sigma_r = \bigl\{\sigma_{k}^{m}(v, x_1, x_2) \mid v \in \nebasr, x_1, x_2 \in A_r \setminus \{a_1\}, k, m > 0\bigr\}\]
    and let \[\abas_r = \bigl\{ \rho_{\varepsilon} \bigr\} \cup
    \bigl\{ \rho_x \mid x \in A_r \setminus \{a_2\} \bigr\} \cup
    \bigl\{ \rho_{x_1a_1^kx_2} \mid x_1, x_2 \in A_r \setminus \{a_2\}, k \geq 0\bigr\}.\]
    Then, $\nbas_r = \Sigma_r \cup \abas_r$ is a basis of $\widehat{\mathcal{C}}_r$. Furthermore,
    for an input function $f$ there exists a representation of $f$ in basis $\nbas_r$
    having size $O(\|f\|)$ that can be obtained in time $O(\|f\|)$.
\end{prop}

The general structure of the proof is the same as in the case of $M_2$.

\begin{lemma} \label{lemm:anindep}
    $\abas_r$ is a linearly independent set.
\end{lemma}
\begin{proof}
    The proof is similar to the proof of Lemma~\ref{lemm:a2indep}. It follows from Theorem~\ref{basis_thm}
    that the set
    \[\abas_r' = \bigl\{ \rho_{\varepsilon} \bigr\} \cup
    \bigl\{ \rho_x \mid x \in A_r \setminus \{a_2\} \bigr\} \cup
    \bigl\{ \rho_{x_1vx_2} \mid x_1, x_2 \in A_r \setminus \{a_2\}, v \in M_r\bigr\}\]
    is a basis of $\widehat{\mathcal{C}}_r$, and $\abas_r \subset \abas_r'$, therefore
    $\abas_r$ is a linearly independent set.
\end{proof}

\begin{lemma} \label{lemm:snindep}
    $\Sigma_r$ is a linearly independent set.
\end{lemma}

\begin{proof}
First, note that the functions $\sigma_k^m(v, x_1, x_2)$ are linearly independent for different
combinations of $v, x_1$ and $x_2$. The proof of this fact is identical to the proof of Lemma~\ref{lemm:sigma_subspaces},
with the difference being the two additional degrees of freedom in the choice of $x_1$ and $x_2$.

\smallskip

Now, let $\Sigma_r^{(v,x_1,x_2)}$ be the set $\{\sigma_k^m(v, x_1, x_2) \mid k, m > 0\}$ where
$v \in \nebasr$ and $x_1, x_2 \in A_r \setminus \{a_1\}$.
Then, what is left to prove the linear independence of the set $\Sigma_r$ is to prove the linear
independence of the set $\Sigma_r^{(v,x_1,x_2)}$ for any fixed word $v$ and letters $x_1, x_2$.
Fortunately, the proof of Lemma~\ref{lemm:sigma_fixed_word} is fully applicable here as each of the
sets $\Sigma_r^{(v,x_1,x_2)}$ also has only two degrees of freedom in $k$ and $m$, and therefore
is representable by a \textit{histogram}. 
\end{proof}

\begin{lemma} \label{lemm:bnindep}
    The span of $\abas_r$ lies in the span of the set
    \[B_r|_{\abas_r} = \bigl\{\rho_{x_1a_1^kx_2} \mid x_1, x_2 \in A_r \setminus \{a_1\}, k \geq 0\bigr\} \cup \bigl\{\rho_{a} \mid a \in A_r \setminus \{a_1\}\bigr\} \cup \bigl\{\rho_{\varepsilon}\bigr\} \subset B_r\]
    and the span of $\Sigma_r$ lies in the span of the set
    \[B_r|_{\Sigma_r} = \bigl\{\rho_{x_1a_1^kva_1^mx_2} \mid v \in \nebasr, x_1, x_2 \in A_r \setminus \{a_1\}, k, m \geq 0\bigr\} \subset B_r.\]
\end{lemma}
\begin{proof}
    The statement about $\abas_r$ follows immediately from the equation \eqref{key_decomposition_degenerate_n} and the statement about $\Sigma_r$ follows immediately from the definition of $\sigma_k^m(v, x_1, x_2)$.
\end{proof}

Similarly to the case of $M_2$, here the sets $B_r|_{\abas_r}$ and $B_r|_{\Sigma_r}$ are also
disjoint subsets of the basis $B_r$, and therefore the set $\nbas_r = \abas_r \cup \Sigma_r$ is also
an independent set.

\medskip

Finally, we show that the set $\nbas_r$ spans $\widehat{\mathcal{C}}_r$.

\begin{lemma} \label{lemm:bnspan}
    The elementary function $\rho_w$ for each word $w \in M_r$ is representable
    as a linear combination of elements of $\nbas_r$.
    Furthermore,
    \begin{enumerate}[label=\arabic*.]
        \item If $w = \varepsilon$ or $w = y$ or $w = y_1a_1^ky_2$ with $k \geq 0$ and $y, y_1, y_2 \in A_r \setminus \{a_2\}$, then $\rho_w \in \abas_r$.
        
        \smallskip
        
        \item If $w = a_2$, then $\rho_w \eqv \rho_{\varepsilon} - \sum_{x \in A_r \setminus \{a_2\}}\rho_{x}$.
        
        \smallskip
        
        \item If $w = a_2a_1^ky$ \ or \ $w = ya_1^ka_2$ with $k \geq 0$ and $y \in A_r \setminus \{a_2\}$ in both cases, then \\ 
        \[\rho_w \eqv \rho_{a_1^ky} - \sum_{x \in A_r \setminus \{a_2\}}\rho_{xa_1^{k}y} \textit{ \, or \, } \rho_w \eqv \rho_{ya_1^k} - \sum_{x \in A_r \setminus \{a_2\}}\rho_{ya_1^{k}x} \quad \text{respectively.}\]
        
        \smallskip
        
        \item If $w = a_2a_1^ka_2$ with $k \geq 0$, then
        \[\rho_w \eqv \rho_{a_1^k}
        - \sum_{x \in A_r \setminus \{a_2\}}\rho_{xa_1^k}
        - \sum_{x \in A_r \setminus \{a_2\}}\rho_{a_1^kx}
        + \sum_{x_1, x_2 \in A_r \setminus \{a_2\}}\rho_{x_1a_1^kx_2}.\]
        
        \smallskip
        
        \item If $w = a_1^kya_1^m$ with $k,m > 0$ and $y \in A_r \setminus \{a_1, a_2\}$, then
        \[\rho_w = -\rho_y + \rho_{a_1^ky} + \rho_{ya_1^m} +
        \sum_{x_1, x_2 \in A_r \setminus \{a_1\}}\sigma_k^m(y, x_1, x_2).\]
        
        \smallskip
        
        \item If $w = a_1^ka_2a_1^m$ with $k,m > 0$, then
        \[\rho_w = -\rho_{\varepsilon} + \sum_{x \in A_r \setminus \{a_2\}}\rho_{x} +\]\[ + \rho_{a_1^k} + \rho_{a_1^m} - \sum_{x \in A_r \setminus \{a_2\}}\rho_{a_1^kx} - \sum_{x \in A_r \setminus \{a_2\}}\rho_{xa_1^m} +
        \sum_{x_1, x_2 \in A_r \setminus \{a_1\}}\sigma_k^m(a_2, x_1, x_2).\]
        
        \smallskip
        
        \item If $w = a_1^ky_1a_1^ny_2a_1^m$ with $k,m,n \geq 0$ and $y_1, y_2 \in A_r \setminus \{a_1\}$, then
        \[
        \rho_w \eqv \rho_{y_1a_1^ny_2} - \]
        \[
         - \sum_{x \in A_r \setminus \{a_2\}}\sigma_k^{n + 1}(y_1, x, y_2)
        + \sum_{x \in A_r \setminus \{a_2\}}\sigma_k^n(y_1, x, y_2) - \]\[
        - \sum_{x \in A_r \setminus \{a_2\}}\sigma_{n + 1}^m(y_2, y_1, x) +
        \sum_{x \in A_r \setminus \{a_2\}}\sigma_n^m(y_2, y_1, x) +\] \[
        + \sum_{x_1, x_2 \in A_r \setminus \{a_1\}}\sigma_k^m(y_1a_1^ny_2, x_1, x_2),
        \]
        where $\rho_{y_1a_1^ny_2}$ is further decomposed by Cases 1, 3 or 4.

        \item Otherwise, $w$ can be written in the form $w = a_1^ky_1a_1^{k'}va_1^{m'}y_2a_1^m$ with $k, m, k', m' \geq 0$,
       $y_1, y_2 \in A_r \setminus \{a_1\}$, and
        $v \in \nebasr$, and
        \[
        \rho_w =
        \sigma_{k' + 1}^{m' + 1}(v, y_1, y_2) - \sigma_{k'}^{m' + 1}(v, y_1, y_2) - \sigma_{k' + 1}^{m'}(v, y_1, y_2) + \sigma_{k'}^{m'}(v, y_1, y_2)
        -\]\[
        - \sum_{x \in A_r \setminus \{a_2\}}\sigma_k^{m' + 1}(y_1a_1^{k'}v, x, y_2)
        + \sum_{x \in A_r \setminus \{a_2\}}\sigma_k^{m'}(y_1a_1^{k'}v, x, y_2)
        -\]\[
        - \sum_{x \in A_r \setminus \{a_2\}}\sigma_{k' + 1}^m(va_1^{m'}y_2, y_1, x)
        + \sum_{x \in A_r \setminus \{a_2\}}\sigma_{k'}^m(va_1^{m'}y_2, y_1, x) +
        \]\[
        + \sum_{x_1, x_2 \in A_r \setminus \{a_1\}}\sigma_k^m(y_1a_1^{k'}va_1^{m'}y_2, x_1, x_2).
        \]
    \end{enumerate}
\end{lemma}

All cases of Lemma~\ref{lemm:bnspan} can be proven similarly to the cases in Lemma~\ref{lemm:b2span}.

\smallskip

We have established that $\nbas_r$ is indeed a basis of $\widehat{\mathcal{C}}_r$, and from Lemma~\ref{lemm:bnspan}
it follows that any elementary counting function $\rho_w$ can be represented in terms of this basis
in a compressed form with the size not exceeding $r^2\|\rho_w\|$. Therefore, if we use an algorithm
similar to the one described in Section \ref{sec:algo} with the cases from Lemma~\ref{lemm:bnspan}
we can obtain a basis representation of any counting function $f$ with the size not exceeding $O(r^2\|f\|)$,
and thus we can check the boundedness of $f$ in $O(r^3\|f\|)$ time for integer coefficients and
$O(r^3\|f\|\log(r\|f\|))$ time for rational coefficients by Lemma~\ref{lemm:normtime}. This concludes
the proof of Theorem~\ref{mainb} in the general case.

\section{Acknowledgements}
The authors are indebted to the anonymous referee for helpful comments on an earlier draft of this paper and a suggested connection to the boundedness problem for weighted automata. 
The work of the second author was prepared within the framework of the HSE University Basic Research Program. Both authors are members of the research group which won the Junior Leader competition of BASIS foundation.



\begin{thebibliography}{12}

\bibitem{abk}
S. Almagor, U. Boker, O. Kupferman,
\emph{What’s Decidable About Weighted Automata?},
Information and Computation, Vol. 282, January 2022, 104651.

\bibitem{bes-fuj}
M. Bestvina and K. Fujiwara,
\emph{Bounded cohomology of subgroups of mapping class groups},
Geom. Topol., 6 (2002), pp. 69--89.

\bibitem{brooks}
R. Brooks, \emph{Some remarks on bounded cohomology}, In: Riemann Surfaces and Related Topics: Proceedings of the 1978 Stony Brook Conference, Annals of Mathematics Studies, Princeton University Press, 1980, pp. 53--63.

\bibitem{clmpw}
W. Czerwiński, E. Lefaucheux, F. Mazowiecki, D. Purser, M.A. Whiteland, \emph{The boundedness and zero isolation problems for weighted automata over nonnegative rationals}, 37th Annual ACM/IEEE Symposium on Logic in Computer Science, LICS 2022, Beer-Sheva, Israel, August 2–5, 2022,  https://doi.org/10.1145/3531130.

\bibitem{cro-ryt}
M. Crochemore and W. Rytter, \emph{Jewels of Stringology: Text Algorithms},
World Scientific, 2002, ISBN 9810248970, 310 pages.

\bibitem{eps-fuji}
D.~B.~A. Epstein and K.~Fujiwara, \emph{The second bounded cohomology of word-hyperbolic groups},
Topology, 36(6), 1997, pp. 1275--1289.

\bibitem{glaz}
J. Glaz, J. Naus, S. Wallenstein, \emph{Scan statistics}, 2009, Springer-Verlag, New York, ISBN 978-0-8176-4748-3, 394 pages.


\bibitem{ggd}
T.~Hartnick and A.~Talambutsa,
\emph{Relations between counting functions on free groups and free monoids}, Groups Geom. Dyn. 12 (2018), no. 4, pp. 1485--1521.

\bibitem{ht-sbornik}
T.~Hartnick and A.~Talambutsa, \emph{Efficient computations with counting functions on free groups and free monoids}, Sbornik: Mathematics, 214(10), 2023, pp. 1458--1499.

\bibitem{grigorch}
R. I. Grigorchuk, \emph{Some results on bounded cohomology}, In: Combinatorial and geometric group theory (Edinburgh, 1993),  London Math. Soc. Lecture Note Ser., 204, Cambridge Univ. Press, Cambridge, 1995, pp. 111--163.

\bibitem{alt}
P.~Kiyashko, \emph{Bases for counting functions on free monoids and groups},
arXiv:2306.15520 [math.GR]

\bibitem{loth}
M. Lothaire, \emph{Combinatorics on words}, 
Cambridge University Press, 2009, ISBN 0521599245, 260 pages.

\bibitem{sapir}
M.V. Sapir, \emph{Combinatorial Algebra: Syntax and Semantics}, Springer Monographs in Mathematics, 2014, ISBN 978-3-319-08030-7, 355 pages.


\end{thebibliography}
\end{document}